\documentclass[12pt]{amsart}

\usepackage{upgreek, amssymb, amsmath}
\usepackage[margin=3.5 cm]{geometry}
\usepackage[pdftex]{graphicx}
\usepackage[shortlabels]{enumitem}
\usepackage{caption,subcaption} 
\usepackage{tikz}
\usepackage{hyperref} 

\newtheorem{theorem}{Theorem}[section]

\newtheorem{prop}{Proposition}[section]

\newtheorem{coro}{Corollary}[section]
\newtheorem{lemma}[theorem]{Lemma}

\newtheorem{obs}{Observation}[section]

\theoremstyle{definition}

\theoremstyle{remark}
\newtheorem{remark}[theorem]{Remark}

\numberwithin{equation}{section}

\newcommand{\comments}[1]{}
\newcommand{\N}{\mathbb{N}}

\newcommand{\R}{\mathbb{R}}
\newcommand{\C}{\mathbb{C}}

\newcommand{\T}{\mathbb{T}}

\DeclareMathOperator{\im}{Im}
\newcommand{\ud}{\mathrm{d}}

\begin{document}

\title[Large coupling asymptotics of the Lyapunov exponent]{Large coupling asymptotics for the Lyapunov exponent of quasi-periodic Schr\"odinger operators with analytic potentials}

\author{Rui Han, C. A. Marx}

\address{Department of Mathematics, UC Irvine, Irvine, CA - 92697}
\email{rhan2@uci.edu}

\address{Department of Mathematics, Oberlin College, Oberlin, OH - 44074}
\email{cmarx@oberlin.edu}





\maketitle

\begin{abstract}
We quantify the coupling asymptotics for the Lyapunov-exponent of a one-frequency quasi-periodic Schr\"odinger operator with analytic potential sampling function. The result refines the well-known lower bound of the Lyapunov-exponent by Sorets and Spencer.
\end{abstract}

\section{Introduction} \label{sec_intro}
Consider a one-dimensional quasi-periodic Schr\"odinger operator
\begin{equation} \label{eq_schrodop}
(H_{\lambda; x} \psi)_n = \psi_{n-1} + \psi_{n+1} + \lambda f(x + n \alpha) \psi_n ~\mbox{,}
\end{equation}
where $\alpha$ is a fixed irrational (``the {\em{frequency}}''), $\lambda \in \mathbb{R}$ denotes the coupling parameter, and $0 \not \equiv f: \mathbb{T} \to \mathbb{R}$ is analytic with extension to a neighborhood of the horizontal strip, 
\begin{equation*}
\mathbb{T}_h:= \{ x + i y ~\vert ~x \in \mathbb{T} ~,~ \vert y \vert \leq h \} ~\mbox{.} 
\end{equation*}

In this paper we are interested in the large coupling behavior of the {\em{Lyapunov exponent}} (LE),
\begin{eqnarray} \label{eq_le}
L(\alpha, A_{\lambda,E}) &= \lim_{n \to \infty} \frac{1}{n} \int_\mathbb{T} \log \Vert A_{\lambda, E}^{(n)}(x; \alpha) \Vert \ud x ~\mbox{,} \\
A_{\lambda, E}^{(n)}(x; \alpha) &= A_{\lambda, E}(x + (n-1)\alpha) \dots A_{\lambda, E}(x) ~\nonumber ~\mbox{,}
\end{eqnarray}
associated with the {\em{Schr\"odinger cocycle}} $(\alpha, A_{\lambda, E})$, 
\begin{equation} \label{eq_schrodcoc}
A_{\lambda, E} = \begin{pmatrix} E - \lambda f & -1 \\ 1 & 0 \end{pmatrix} ~\mbox{.}
\end{equation}
Schr\"odinger cocycles are special cases of general analytic cocycles $(\alpha, D)$ induced by a given analytic matrix-valued function $D: \mathbb{T} \to M_2(\mathbb{C})$ and a fixed $\alpha \in \mathbb{T}$. The cocycle $(\alpha, D)$ defines a dynamical system on $\mathbb{T} \times \mathbb{C}^2$ given by the linear skew-product $(x,v) \mapsto (x + \alpha, D(x) v)$. The Lyapunov exponent $L(\alpha, D) \in [-\infty, +\infty)$ of the cocycle $(\alpha, D)$ is then defined in analogy to (\ref{eq_le}). We note that since $A_{\lambda,E} \in SL(2,\mathbb{R})$, one necessarily has $L(\alpha, A_{\lambda,E}) \geq 0$.

It has been known since the pioneering works of Herman \cite{Herman_1983} and Sorets-Spencer \cite{SoretsSpencer_1991} that for $\vert \lambda \vert \geq \lambda_0$, the LE is positive as a consequence of the lower bound 
\begin{equation} \label{eq_soretsspencerbound}
L(\alpha, A_{\lambda, E}) \geq \log \vert \lambda \vert - c ~\mbox{.}
\end{equation}
Here, $c$ and $\lambda_0$ are independent of $\alpha$ and only depend on the properties of $f$. In particular, (\ref{eq_soretsspencerbound}) rigorously establishes the heuristics that for sufficiently large $\lambda$, the potential in (\ref{eq_schrodop}) should dominate the discrete Laplacian. We mention that simplified proofs of the Sorets-Spencer bound (\ref{eq_soretsspencerbound}) were obtained in \cite{Bourgain_book_2005}, \cite{Zhang_2012}, and \cite{KleinDuarte_2013_posLE}.

The question of capturing the large coupling behavior of the LE has a rich history which influenced greatly the development of many aspects of the theory of quasi-periodic operators. For a more extensive, recent account of both the history and some of the key techniques emerging from the work on this question we refer the reader to \cite{JitomirskayaMarx_ETDS_2016_review} and \cite{Damanik_review2014}. 

The main result of this paper refines the Sorets-Spencer bound in (\ref{eq_soretsspencerbound}) by establishing a {\em{rate of convergence}} in the coupling. Following, $\mathcal{C}_h^\omega(\mathbb{T}; \mathbb{R})$ denotes the one-periodic, real analytic functions with extension to a neighborhood of $\mathbb{T}_h$ and equipped with the norm, $\Vert f \Vert_h:= \sup_{z \in \mathbb{T}_h} \vert f(z) \vert$. More generally, $\mathcal{C}_h^\omega(\mathbb{T})$ will denote the complex analytic, one-periodic functions with extension to a neighborhood of $\mathbb{T}_h$.

\begin{theorem} \label{thm_main}
Given a quasi-periodic Schr\"odinger operator (\ref{eq_schrodop}) with $f \in \mathcal{C}_h^\omega(\mathbb{T}; \mathbb{R})$ and $\alpha$ irrational. For every $0<\rho < \frac{1}{2}\min{(h,1)}$, there exist $0 < \lambda_0=\lambda_0(f, \rho)$, $N=N(f,\rho) \in \mathbb{N}$ and $0 < C=C(f,\rho)$ such that for all , $|\lambda| > \lambda_0$, and $E \in \mathbb{R}$, one has
\begin{align} \label{eq_asymptoticform}
L(\alpha, A_{\lambda, E})=\ln{|\lambda|}+\int_{\mathbb{T}} \log \vert {E}/{\lambda} - f(x) \vert ~\ud x+O(|\lambda|^{-\frac{2}{2N+1}}),
\end{align}
where $|O(|\lambda|^{-\frac{2}{2N+1}})|\leq C|\lambda|^{-\frac{2}{2N+1}}$. 
\end{theorem}
\begin{remark} \label{remark_quantification}
The proof determines $N=N(f,\rho)$ in Theorem \ref{thm_main} in terms of the {\em{number of zeros}} of the family $\{ f - E/\lambda\}_{E/\lambda \in \mathbb{R}}$ on its domain of analyticity. Specifically, for $g \in \mathcal{C}_h^\omega(\mathbb{T}; \mathbb{R})$ and $0< \epsilon < h$, denote by $N_{[-\epsilon, \epsilon]}(g)$ the number of zeros of $g$ on $\mathbb{T}_\epsilon$ counting multiplicities. Given $f \in \mathcal{C}_h^\omega(\mathbb{T}; \mathbb{R})$ as in Theorem \ref{thm_main}, we let
\begin{equation}\label{defN}
\hat{N}_\epsilon(f):= \max_{\mu \in \mathbb{R}} N_{[-\epsilon,\epsilon]}(f - \mu) ~\mbox{.}
\end{equation}
We then show that $N=N(f,\rho)$ in Theorem \ref{thm_main} can be chosen as
\begin{equation} \label{eq_zerocountingaccel}
N = \hat{N}_{2\rho}(f) ~\mbox{,}
\end{equation}
in which case, 
\begin{equation} \label{eq_lambda0}
\lambda_0 = 2^{2N + 1} ~\mbox{,}
\end{equation}
and $C$ is as quantified in (\ref{errorestimate}).
\end{remark}

This article was motivated by a recent work by Duarte-Klein \cite{DuarteKlein_2015_contising} and an announced result by Ge-You \cite{GeYou}. 
In \cite{DuarteKlein_2015_contising} an asymptotic formula similar to (\ref{eq_asymptoticform}) was obtained however with a slower, weak-H\"older type rate of convergence $O(\mathrm{e}^{-c (\log \vert \lambda \vert)^b})$, for constants $0 < c$ and $0 < b < 1$ (see Proposition 1.1 in \cite{DuarteKlein_2015_contising}). 
Moreover, the strategy of their proof requires the authors to impose a {\em{Diophantine}} condition on $\alpha$, which lead to a dependence of $\lambda_0$ on the Diophantine properties of the frequency. 
In \cite{GeYou}, a sharper rate of convergence $O(|\lambda|^{-1})$ was obtained, however only for {\em trigonometric polynomial} $f$.
To our knowledge, the authors reduced estimating the Lyapunov exponent of (\ref{eq_schrodop}) to that of its dual Hamiltonian. 
$f$ being trigonometric polynomial ensures the dual Hamiltonian to be finite-range, thus almost reducibility methods are applicable.  

Theorem \ref{thm_main} strengthens the result of Duarte-Klein by improving the rate of convergence from weak-H\"older to H\"older, as well as extending the validity of the asymptotic formula to cover {\em{all}} irrational frequencies, irrespective of the Diophantine properties.

We do however point out that the result in \cite{DuarteKlein_2015_contising} was proven in the more general setting of {\em{multi-frequency}} Schr\"odinger cocycles, where $f$ in (\ref{eq_schrodcoc}) is analytic on $\mathbb{T}^d$ and the base dynamics is given by translations by a fixed Diophantine frequency-vector $\alpha \in \mathbb{T}^d$. Since our proof relies on the property that the zeros of a non-identically vanishing, analytic function in one variable are isolated, our work does {\em{not}} extend to the multi-frequency case.

\subsection{Outline of the proof of Theorem \ref{thm_main}} \label{sec_outline}

The key idea underlying much of the work on positivity of the LE is to argue that for sufficiently large coupling the upper left entry in (\ref{eq_schrodcoc}) determines the dynamics of $(\alpha, A_{\lambda, E})$. In \cite{DuarteKlein_2015_contising} this is achieved by factoring the cocycle (\ref{eq_schrodcoc}) according to
\begin{equation} \label{eq_factoriz_dk}
A_{\lambda, E} = (\lambda) \begin{pmatrix} E/\lambda - f  & o(1) \\ o(1) & 0 \end{pmatrix} ~\mbox{, for $\lambda \to \infty$.}
\end{equation}
Morally this reduces the problem to the ``limiting cocycle'' induced by $D_\infty = (\begin{smallmatrix} E/\lambda - f & 0 \\ 0 & 0 \end{smallmatrix})$ whose LE is simply given by $\int_\mathbb{T} \log \vert f(x) - E/\lambda \vert ~\ud x$, thereby already accounting for the first two terms in (\ref{eq_asymptoticform}). Of course this reduction to the ``limiting cocycle'' requires establishing a continuity theorem which quantifies the modulus of continuity of the LE locally about $(\alpha, D_\infty)$. The primary technical obstacle in this context is however provided by the everywhere non-invertibility (or {\em{singularity}}) of the cocycle $(\alpha, D_\infty)$ as $\det D_\infty(x) \equiv 0$.

Duarte-Klein overcome this technical problem by proving that for {\em{Diophantine}} $\alpha \in \mathbb{T}^d$, the LE is {\em{weakly H\"older continuous}} at $(\alpha, D_\infty)$, i.e. characterized by the modulus of continuity $\omega(h) = C \mathrm{e}^{-c (\log \vert \lambda \vert)^b}$ for constants $C,c>0$ and $0 < b < 1$ ($b = 1$ would correspond to H\"older continuity) which depend on the Diophantine properties of $\alpha$. The latter then yields (\ref{eq_asymptoticform}) with however the H\"older-rate replaced by a weak-H\"older rate. We mention that Duarte-Klein obtain this continuity result as a special case of a general quantitative continuity theorem for the LE of singular analytic cocycles which, in fact, constitutes the main result of \cite{DuarteKlein_2015_contising}; see also \cite{DuarteKlein_monograph}.

Instead of employing the factorization in (\ref{eq_factoriz_dk}), we follow the approach already present in the original work of Sorets-Spencer, which takes advantage of the simple structure of the zero set of analytic functions in {\em{one}} variable. 

Extending the cocycle to the complex plane, $A_{E, \lambda; y}(x) :=A_{E, \lambda}(x + i y)$, isolatedness of zeros implies that $(E/\lambda - f(\cdot + i y))$ is bounded away from zero for $0 < \vert y \vert < h$, which in turn allows to factorize according to
\begin{equation} \label{eq_factoriz_trad}
A_{\lambda, E; y}(x) = \lambda (E/\lambda - f(x + i y)) \begin{pmatrix} 1 & o(1) \\ o(1) & 0 \end{pmatrix} ~\mbox{, for $\lambda \to \infty$.}
\end{equation}
Even though the limiting co-cycle $D_\infty^\prime = (\begin{smallmatrix} 1 & 0 \\ 0 & 0 \end{smallmatrix})$ is still singular, the latter induces extremely regular dynamics known as a {\em{dominated splitting}} (for a definition, see Sec. \ref{sec_keylemmas}). Locally about dominated splittings the modulus of continuity of the LE is significantly stronger than weakly H\"older, which we quantify explicitly in Lemma \ref{lem_ds}. The latter already yields a version of our asymptotic formula (\ref{eq_asymptoticform}) {\em{off}} the real axis, which, using that $L(\alpha, A_{\lambda, E;y})$ is known to be {\em{convex}} in $y$, will be shown to carry over to $y = 0$.

\subsection{Consequences for the acceleration} \label{sec_intro_accel}
Finally, since our proof of (\ref{eq_asymptoticform}) determines the LE of the {\em{phase-complexified}} Schr\"odinger cocycle $L(\alpha, A_{\lambda, E;y})$ for $y > 0$, it also allows us to extract estimates for its right-derivative,
\begin{equation} \label{eq_acceldefn}
\omega(\alpha, \lambda, E; y):= \frac{1}{2 \pi} \lim_{t \to 0^+} \dfrac{L(\alpha, A_{\lambda, E;y + t}) - L(\alpha, A_{\lambda, E;y})}{t} \in \mathbb{Z} ~\mbox{.}
\end{equation}
The right-derivative (\ref{eq_acceldefn}), which exists by convexity of $y \mapsto L(\alpha, A_{\lambda, E;y})$, is known as the {\em{acceleration}} and was first introduced by Artur Avila in his global theory of analytic one-frequency Schr\"odinger operators \cite{Avila_globalthy_published}. In \cite{Avila_globalthy_published}, Avila shows that $\omega(\alpha, \lambda, E; y) \in \mathbb{Z}$ (``{\em{quantization of the acceleration}}''), which in particular allows for a stratification of the spectrum according to the values of the acceleration at $y = 0$. Avila then relates the values of the acceleration at $y = 0$ to the spectral properties of the operator (\ref{eq_schrodop}). The acceleration at $y = 0$ thus forms the key to Avila's dynamical formulation of the spectral theory of quasi-periodic Schr\"odinger operators; consequently, various conjectures about its behavior have been made, see e.g. Sec. 2.1.4 in \cite{Avila_globalthy_published}.

As a corollary to the proof of Theorem \ref{thm_main}, we extract the following upper bound of the acceleration at $y = 0$ in terms of the zero-counting function $N=N(f,\rho)$ defined in (\ref{eq_zerocountingaccel}):
\begin{coro} \label{coro_accelestim}
Let $f \in \mathcal{C}_h^\omega(\mathbb{T}; \mathbb{R})$, $0<\rho<\frac{1}{2}\min{(h,1)}$, as in Theorem \ref{thm_main}, and $N=N(f,\rho)$ and $\lambda_0= \lambda_0(f, \rho)$ as given in (\ref{eq_zerocountingaccel})-(\ref{eq_lambda0}). Then, for all $\vert \lambda \vert \geq \lambda_0$ and $E \in \mathbb{R}$, the acceleration at $y = 0$ is bounded above by
\begin{equation} \label{eq_estimaccel}
\omega(\alpha, \lambda, E; 0) \leq \sup_{y \in [0, \rho]} \frac{1}{2 \pi} \dfrac{\partial}{\partial y} \int_\mathbb{T} \log \vert E/\lambda - f(x + i y) \vert ~\ud x \leq \frac{1}{2} N ~\mbox{.}
\end{equation}
\end{coro}
\begin{remark}
Introducing a stratification of the energy axis described in Sec. \ref{sec_improvements}, we are even able to improve on the estimate given in (\ref{eq_estimaccel}): For each of these strata, we will argue that the acceleration at $y=0$ can simply be bounded above by the number of intersection points of the graph of $f$ with horizontal lines, see Theorem \ref{prop_strengthned} for details.
\end{remark}

We structure the remainder of the paper as follows. Sec. \ref{sec_keylemmas} contains the two key ingredients that were outlined in Sec. \ref{sec_intro_accel} above: Lemma \ref{DK} gives a quantitative lower bound of $\vert E/\lambda - f(x + i y) \vert$, uniformly both in $x \in \mathbb{T}$ and $E/\lambda \in \mathbb{R}$, for some positive imaginary part $0< y$ which in general depends on the value of the parameter $E/\lambda$. Lemma \ref{lem_ds} quantifies the modulus of continuity of the LE  locally about the limiting cocycle $(\alpha, D_\infty)$. Sec. \ref{app_keyfunctional} discusses some technical properties that underlie the estimate of the acceleration given in Corollary  \ref{coro_accelestim}. 

Finally, the proof of the main results, Theorem \ref{thm_main} and Corollary \ref{coro_accelestim}, is presented in Sec. \ref{sec_proofthmmain}. To conclude, Sec. \ref{sec_improvements} discusses the above-mentioned improvements of the upper bound on the acceleration through a stratification of the energy-axis (Theorem \ref{prop_strengthned}).

\section{Key lemmas} \label{sec_keylemmas}

Our proof of Theorem \ref{thm_main} relies on two key lemmas. The first of the two implies that the factorization in (\ref{eq_factoriz_trad}) can be done within the analytic category. It is an immediate consequence of the simple structure of the zero set of analytic functions of a single variable. 

A non-quantitative form already played a crucial role in the original work of Sorets-Spencer (see also Ch. 3 in \cite{Bourgain_book_2005}):
\begin{obs} \label{obs_key}
Let $0 \not \equiv f \in \mathcal{C}_h^\omega(\mathbb{T})$. For every $0 < \delta < h$, there exists $\eta > 0$ such that
\begin{equation*}
\min_{\mu \in \mathbb{R}} \max_{\delta/2 \leq y \leq \delta} \min_{x \in \mathbb{T}} \vert f(x + i y) - \mu \vert > \eta ~\mbox{,}
\end{equation*}
\end{obs}
\begin{proof}
Assuming that such $\eta > 0$ did not exist, would imply that \break $\max_{\delta/2 \leq y \leq \delta} \min_{x \in \mathbb{T}} \vert f(x + i y) - \mu_0 \vert = 0$, for some $\mu_0 \in \mathbb{R}$. The latter however contradicts isolatedness of zeros for non-identically vanishing analytic functions.
\end{proof}
A quantitative version of Observation \ref{obs_key} was given in \cite{KleinDuarte_2013_posLE}, see Corollary 4.6 therein. To formulate it, for an interval $I\subseteq [-h,h]$, let $N_I(g)$ be the number of zeros of $g$ on $\{x+iy\ |\ x\in \T, y\in I\}$, counting multiplicities. 
For $0<\epsilon<h$, we denote $N_{\epsilon}(g):=N_{[-\epsilon, \epsilon]}(g)$ and $\tilde{N}_{\epsilon}(g):=N_{(-\epsilon, \epsilon)}(g)$ for simplicity. 
Let $\hat{N}_\epsilon(f)$ be as in (\ref{defN}) and define
\begin{align*}
\hat{\beta}_\epsilon(f):&= \min_{\mu\in \mathbb{R}} \beta_\epsilon(f - \mu) ~\mbox{,} \\
\beta_\epsilon(f - \mu):&= \min_{z \in \mathbb{T}_{\epsilon}} \vert g_{2\epsilon}(f - \mu)(z) \vert ~\mbox{,}
\end{align*}
where $g_{2\epsilon}(f-\mu)(z)$ is the {\em{zero-free part}} of $(2e^{2\pi}+2)^{\tilde{N}_{2\epsilon}(f-\mu)}
(f(z) - \mu)$ on $\mathrm{int}(\mathbb{T}_{2\epsilon})$.
It is shown in \cite{KleinDuarte_2013_posLE} that $N_{\epsilon}(f-\mu)$ is upper semi-continuous and $\beta_{\epsilon}(f-\mu)$ is lower semi-continuous in $\mu$, in particular $\hat{N}_{\epsilon}(f)<\infty$ and $\hat{\beta}_{\epsilon}(f)>0$. 
The quantitative version of Observation \ref{obs_key} forms our first key lemma: 
\begin{lemma}[Duarte-Klein \cite{KleinDuarte_2013_posLE}]\label{DK}
For any $0<\delta<\rho<\frac{1}{2}\min{(h,1)}$, one has
\begin{align}
\min_{\mu\in \R}\max_{\frac{\delta}{2}\leq y\leq \delta} \min_{x\in \mathbb{T}} |f(x+iy)-\mu|\geq \hat{\beta}_{\rho}(f)\left(\frac{K_1 \delta}{\hat{N}_{2\rho}(f)}\right)^{\hat{N}_{2\rho}(f)},
\end{align}
where $K_1$ is an absolute constant quantified in (\ref{eq_constantk1}).
\end{lemma}
To keep the paper self-contained, we include a short proof of Lemma \ref{DK} in Appendix \ref{app_DuarteKlein}.

The second key lemma guarantees that the factorization in (\ref{eq_factoriz_trad}) generates very regular dynamics in the sense of a dominated splitting, a notion originating from the theory of partially hyperbolic dynamical systems: Given an analytic cocycle $(\alpha, D)$, it is said to induce a {\em{dominated splitting}} if there exists a continuous, nontrivial splitting $\C^2=E_x^{(1)}\oplus E_x^{(2)}$ and $N\in\N$ satisfying
\begin{enumerate}
\item[(i)] $D^{(N)}(x; \alpha) E_x^{(j)}\subseteq E_{x+N\alpha}^{(j)}$, for $1 \leq j \leq 2$,
\item[(ii)] for all $v_j\in E_x^{(j)}\setminus\{0\}$, $1 \leq j \leq 2$, one has
\begin{equation*}
\frac{\Vert D^{(N)}(x; \alpha) v_1\Vert}{\Vert v_1\Vert}>\frac{\Vert D^{(N)}(x; \alpha)v_2\Vert}{\Vert v_2\Vert} ~\mbox{.} 
\end{equation*}
\end{enumerate}
Here, as before, we let $D^{(N)}(x ; \alpha) = \prod_{j = N-1}^{0} D(x + j \alpha)$.

From a general point of view, a theorem by Ruelle \cite{Ruelle_1979_LEanalytic} already guarantees that dominated splitting is an open property in the cocycle and that the LE is locally {\em{smooth}} about dominated splittings. In view of Theorem \ref{thm_main}, we however need a more {\em{quantitative}} version of this result, which we proved in \cite{Marx_Shou_Wellens_2015}, see Proposition 3.1, therein:
\begin{lemma}[Marx-Shou-Wellens \cite{Marx_Shou_Wellens_2015}] \label{lem_ds}
Given $D= (\begin{smallmatrix} 1 & g^{-1} \\ g^{-1} & 0 \end{smallmatrix}) \in \mathcal{C}(\mathbb{T}; M_2(\mathbb{C}))$ and $\alpha \in \mathbb{T}$. If 
\begin{equation*}
m(g):= \inf_{x \in \mathbb{T}} \vert g(x) \vert > 2 ~\mbox{,}
\end{equation*}
then the cocycle $(\alpha, D)$ induces a dominated splitting and its Lyapunov exponent satisfies
\begin{eqnarray} \label{eq_complexLE_ds_bd_1}
~\quad & ~~ \log \left\{ \dfrac{ (1 - \frac{\sigma(g)}{m(g)})^2 + \frac{1}{m(g)^2} }{1 + \sigma(g)^2}  \right\} \leq 2 L(\alpha, D) \leq \log \left\{ \left(1 + \frac{\sigma(g)}{m(g)}\right)^2 + \frac{1}{m(g)^2} \right\}     ~\mbox{, } \nonumber \\ 
& \sigma(g) = \min\left\{ 1 ~;~ \dfrac{m(g)-1}{m(g)(m(g)-2)} \right\} \label{eq_estimatesolutionFP} ~\mbox{.}
\end{eqnarray}
\end{lemma}
\begin{remark} \label{rem_ds}
As shown in \cite{Marx_Shou_Wellens_2015} (see Remark 3.2, therein), the lower bound $m(g) > 2$ in Lemma \ref{lem_ds} is in general {\em{optimal}}.
Note that from (\ref{eq_estimatesolutionFP}), if $m(g) > 2$, $\sigma(g) \asymp m(g)^{-1}$ yields
\begin{equation} \label{eq_LEquantitativemodulus}
\frac{-K_2}{m(g)^2} \leq L(\alpha, D) \leq \frac{K_3}{m(g)^2} ~\mbox{,}
\end{equation}
Specifically, using (\ref{eq_complexLE_ds_bd_1}), we extract numerical values for the constants $K_2, K_3$, given by
\begin{equation} \label{eq_quantificConst}
K_2 \approx 8.4985 ~\mbox{, } K_3 \approx 6.5451 ~\mbox{,}
\end{equation}
for details, see Appendix \ref{app_numerical}.
\end{remark}

\section{Acceleration and winding number}\label{app_keyfunctional}

As mentioned in Sec. \ref{sec_outline}, the asymptotic formula for the LE in (\ref{eq_asymptoticform}) will first be established away from the real axis by extending the Schr\"odinger cocycle to the complex plane. The extrapolation to zero imaginary part as well as Corollary \ref{coro_accelestim} will rely on the properties of functionals of the form
\begin{equation} \label{eq_babyaccel}
I[f](y):=  \int_\mathbb{T} \log \vert f(x + i y) \vert ~\ud x ~\mbox{,}
\end{equation}
the discussion of which forms the purpose of this section; following $f \in \mathcal{C}_h^\omega(\mathbb{T})$ denotes an arbitrary complex analytic function which does not vanish identically. 

It is easy to see that $y \mapsto I[f](y)$ is convex. In particular, the right-derivative
\begin{equation}\label{defomega}
\omega[f](y):= \frac{1}{2 \pi} D_+ I[f](y) ~\mbox{,}
\end{equation}
is a well-defined, right-continuous, increasing function in $y$. We note that (\ref{eq_babyaccel}) can in fact be viewed as the LE of the cocycle $(\alpha, D_y)$ where $D_y = (\begin{smallmatrix} f(\cdot + i y) & 0 \\ 0 & 0 \end{smallmatrix})$, in which case (\ref{defomega}) is its acceleration. We will be particularly interested in the situation when $f$ is {\em{real}}-analytic, in which case $I[f](y)$ is an even and $\omega[f](y)$ is an odd function of $y$.

The following proposition can be considered an extension of Jensen's formula, versions of which also played an important role in the original proof of the Sorets-Spencer's lower bound. 
\begin{prop} \label{prop_accelvalue}
Let $g$ denote the zero-free part of $f$ on $\mathbb{T}_h$. For all $y \in [-h,h]$, one has
\begin{equation} \label{eq_prop_accelvalue}
\omega[f](y) =  -\mathrm{ind}_{g(\cdot + i0)}(0) -N_{(y,h]}(f) ~\mbox{,}
\end{equation}
where $\mathrm{ind}_{g(\cdot+i0)}(0)$ is the winding number of $g(\cdot+i0)$ around $0$.
In particular, if $f$ is real-analytic, one has 
\begin{equation}  \label{eq_prop_accelvalue_real}
\omega[f](y) = \frac{1}{2} N_{(-y, y]}(f) ~\mbox{,  $0 \leq y \leq h$ .}
\end{equation}
\end{prop}

Proposition \ref{prop_accelvalue} could be extracted from Lemma 1 in \cite{SoretsSpencer_1991}. For convenience of the reader, we include a short, alternative proof below. The argument is based on the following simple observation, which will be useful in its own right:
\begin{obs} \label{obs_indexaccel}
Let $f \in \mathcal{C}_h^\omega(\mathbb{T})$. Suppose for $-h \leq y_1 < y_2 \leq h$ that $\inf_{y_1 \leq \im(z) \leq y_2} \vert f(z) \vert >0$. Then for all $y \in [y_1, y_2]$,
\begin{equation*}
\omega[f](y) = - \mathrm{ind}_{f(\cdot+ i y_1)}(0) ~\mbox{.}
\end{equation*}
\end{obs}
\begin{proof}[Proof of Proposition \ref{prop_accelvalue}]
Let $z_k$, $1 \leq k \leq N_h(f)$, be the zeros of $f$ on $\mathbb{T}_h$, i.e.,
\begin{equation*}
f(z) = g(z) \prod_{k=1}^{N_h(f)} (e^{2\pi iz} - e^{2\pi iz_k}) ~\mbox{, $z \in \mathbb{T}_h$ .}
\end{equation*}
Thus, if $z = x + i y \in \mathbb{T}_h$ with $y \neq \im z_k$, for all $1 \leq k \leq N_h(f)$, Observation \ref{obs_indexaccel} implies
\begin{equation*}
\omega[f](y) = - \mathrm{ind}_{g(\cdot + i 0)}(0) - \sum_{k=1}^{N_h(f)} \mathrm{ind}_{\mathrm{e}^{2 \pi i \cdot}\mathrm{e}^{-2 \pi y} - e^{2\pi iz_k}}(0) ~\mbox{.}
\end{equation*}
Since
\begin{equation*}
\mathrm{ind}_{\mathrm{e}^{2 \pi i \cdot}\mathrm{e}^{-2 \pi y} - e^{2\pi iz_k}}(0) = \begin{cases} 1 & ~\mbox{, } y < \im z_k ~\mbox{,} \\ 0 & ~\mbox{, } y \geq \im z_k ~\mbox{,} \end{cases}
\end{equation*}
(\ref{eq_prop_accelvalue}) follows by right-continuity of $\omega[f](y)$ and $N_{(y, h]}(f)$.

Finally, if $f$ is real-analytic, oddness of $\omega[f](y)$ allows to conclude
\begin{equation*}
- \mathrm{ind}_{g(. + i 0)}(0) = \frac{1}{2} ( N_{(-y, h]}(f) + N_{(y, h]}(f) ) ~\mbox{,}
\end{equation*}
whence, for $0 \leq y \leq h$, (\ref{eq_prop_accelvalue}) implies
\begin{equation*}
\omega[f](y) = \frac{1}{2} ( N_{(-y,h]}(f) - N_{(y, h]}(f) ) ~\mbox{,}
\end{equation*}
which reduces to (\ref{eq_prop_accelvalue_real}).
\end{proof}

\section{Proof of Theorem \ref{thm_main}} \label{sec_proofthmmain}

Fix $0<\rho<h$, let $N=\hat{N}_{{2\rho}}(f)$ and $\lambda_0=2^{2N+1}$. 
Take $|\lambda| > \lambda_0$ and $\delta={N}K_1^{-1} \hat{\beta}_{\rho}(f)^{{-1/N}} |\lambda|^{-\frac{2}{2N+1}}$.

By Lemma \ref{DK}, for any $E\in \mathbb{R}$ and $|\lambda|>\lambda_0$, we have
\begin{align}\label{lambdadelta1}
\min_{x\in \mathbb{T}} |\lambda f(x+iy)-E|\geq \hat{\beta}_{\rho}(f) |\lambda| (\frac{K_1 \delta}{N})^{N}=|\lambda|^{\frac{1}{2N+1}}>2,
\end{align}
for some $y=y(\lambda, E, f, \rho, h)\in [\frac{\delta}{2}, \delta]$. Continuity of $\min_{x\in \mathbb{T}} |\lambda f(x+iy)-E|$ in $y$ then shows that in fact
\begin{align} \label{y1y2}
\min_{x\in \mathbb{T}} |\lambda f(x+iy)-E|>2 ~\mbox{, for all $y\in [y_1, y_2]$,}
\end{align}
where $[y_1, y_2]$ is some sub-interval of $[{\delta}/{2}, \delta]$.

For $y \in [y_1, y_2]$ we can thus factorize
\begin{equation} \label{eq_LEOFFreal}
A_{\lambda, E; y}(x): = \lambda (\frac{E}{\lambda} - f(x + i y)) \underbrace{\begin{pmatrix} 1 & -\frac{1}{\lambda} (\frac{E}{\lambda} - f(x + i y))^{-1} \\ \frac{1}{\lambda} (\frac{E}{\lambda} - f(x + i y))^{-1} & 0   \end{pmatrix}}_{=: D_{\lambda, E; y}(x)}  ~\mbox{,}
\end{equation}
which allows the LE to be written in the form,
\begin{equation} \label{eq_fact_le}
L(\alpha, A_{\lambda, E; y}) = \ln{|\lambda|}+ \int_\mathbb{T} \log \vert {E}/{\lambda} - f(x + i y) \vert ~\ud x + L(\alpha, D_{\lambda, E; y}) ~\mbox{.}
\end{equation}

Using (\ref{y1y2}), we may apply Lemma \ref{lem_ds} and Remark \ref{rem_ds} to the cocycle $(\alpha, D_{\lambda, E; y})$, which implies that $(\alpha, D_{\lambda, E; y})$ induces a dominated splitting with 
\begin{equation}\label{controlD}
|L(\alpha, D_{\lambda, E; y})| \leq K_2 |\lambda|^{-\frac{2}{2N+1}} ~\mbox{.}
\end{equation}
where $K_2$ is an absolute constant quantified in (\ref{eq_quantificConst}).

Moreover, since the LE is known to be smooth about dominated splittings \cite{Ruelle_1979_LEanalytic} and $D_{\lambda, E; y} \to (\begin{smallmatrix} 1 & 0 \\ 0 & 0 \end{smallmatrix})$ as $\lambda\rightarrow\infty$, constancy of the limiting cocycle $(\alpha, (\begin{smallmatrix} 1 & 0 \\ 0 & 0 \end{smallmatrix}))$ and quantization of the acceleration shows that $(\alpha, D_{\lambda, E; y})$ itself has zero acceleration. We emphasize that quantization of the acceleration is indeed necessary here to conclude zero acceleration. Smoothness of the LE in the regime of dominated splittings ($\mathcal{DS}$) already yields continuity of the acceleration of $(\alpha, D_{\lambda, E; y})$ once $\lambda$ is sufficiently large so that Lemma \ref{lem_ds} applies (i.e. for $\vert \lambda \vert \geq \lambda_0$). Knowing, in addition, that the acceleration is quantized, implies that it is in fact {\em{constant}} and thus equal to the acceleration of $(\alpha, (\begin{smallmatrix} 1 & 0 \\ 0 & 0 \end{smallmatrix}))$ for all $\vert \lambda \vert \geq \lambda_0$.

Thus, we conclude from (\ref{defomega}), (\ref{y1y2}) and (\ref{eq_fact_le}) that one has
\begin{eqnarray} \label{eq_1}
\omega(\alpha, E, \lambda; y) = \omega[{E}/{\lambda}-f](y)~\mbox{,}
\end{eqnarray}
for all $y_1 \leq y < y_2$. 

Since $L(\alpha, A_{\lambda, E, y})$ is even and convex in $y$, we get
\begin{align}\label{ac1}
0\leq \omega(\alpha, \lambda, E;0)\leq \omega[{E}/{\lambda}-f](y),
\end{align}
and
\begin{align}\label{Ly0difference}
|L(\alpha, A_{\lambda, E})-L(\alpha, A_{\lambda, E; y})|\leq 2\pi \omega [{E}/{\lambda}-f](y)\ |y|.
\end{align}
Note that by (\ref{eq_prop_accelvalue_real}), $\omega[{E}/{\lambda}-f](t)\leq \frac{1}{2} N$ for any $0<t<\rho$, hence we conclude the upper bound on the accerleration given in Corollary \ref{coro_accelestim}.
 
Finally, combining (\ref{eq_fact_le}), (\ref{controlD}) with (\ref{Ly0difference}) yields
\begin{align*}
&|L(\alpha, A_{\lambda, E})-\log{|\lambda|}-\int_{\mathbb{T}} \log \vert {E}/{\lambda} - f(x) \vert ~\ud x | \notag\\ 
\leq &|L(\alpha, A_{\lambda,E})-L(\alpha, A_{\lambda, E;y})|+|L(\alpha, A_{\lambda, E;y})-\ln{|\lambda|}-\int_{\mathbb{T}} \log \vert {E}/{\lambda} - f(x) \vert ~\ud x| \notag\\
\leq & N\pi |y|+\left|\int_{\mathbb{T}} \log \vert {E}/{\lambda} - f(x) \vert ~\ud x-\int_{\mathbb{T}} \log \vert {E}/{\lambda} - f(x+iy) \vert ~\ud x\right| + K_2 \lambda^{-\frac{2}{2N+1}} \notag\\
=&  N\pi |y|+2\pi \left|\int_{0}^{y} \omega[{E}/{\lambda}-f](t) ~\ud t\right|+ K_2 |\lambda|^{-\frac{2}{2N+1}} \notag\\
\leq &2N\pi  |y|+ K_2 |\lambda|^{-\frac{2}{2N+1}} \notag\\
\leq &(2N^2 K_1^{-1}\hat{\beta}_{\rho}(f)^{-1/N} \pi+K_2) |\lambda|^{-\frac{2}{2N+1}}~\mbox{,}
\end{align*}
which completes the proof of Theorem \ref{thm_main} with the constant $C(f, \rho)$ given by
\begin{equation} \label{errorestimate}
C(f, \rho) = 2N^2 K_1^{-1}\hat{\beta}_{\rho}(f)^{-1/N} \pi+K_2 ~\mbox{.}
\end{equation}

\section{Stratified estimates of the accerleration} \label{sec_improvements}

This final section addresses improvements of the estimate of the acceleration given in Corollary \ref{coro_accelestim}. Whereas this upper bound was merely used as a tool to extrapolate the asymptotic formula first proven for $y>0$ (see (\ref{eq_fact_le}) and (\ref{controlD})) to $y=0$, estimates on the acceleration are interesting in their own right from the point of view of Avila's global theory, see the comments in Sec. \ref{sec_intro_accel}.

If one is only interested in the complexified LE on the spectrum $\Sigma(\lambda)$ of $H_{\lambda;x}$, the bound
\begin{equation*}
\Sigma(\lambda) \subseteq [-2 +\lambda \min_{x \in \mathbb{T}} f(x), 2 + \lambda \max_{x \in \mathbb{T}} f(x)] ~\mbox{ }
\end{equation*}
allows to restrict the values of $\mu =E/ \lambda$ for $\lambda>0$ to the compact set $[-2/\lambda +\min_{x \in \mathbb{T}} f(x), 2/\lambda + \max_{x \in \mathbb{T}} f(x)]$.

Let $\mathcal{Z}(f^\prime)=\{x\in \T\ |\ f^\prime(x)=0\}$ be the set of critical points and denote
\begin{equation}
\mathcal{D}(f):=\{f(x)\ |\ x\in \mathcal{Z}(f^\prime)\} ~\mbox{.}
\end{equation}
Note that for $[\mu_1, \mu_2]\subset [\min_{x\in \mathbb{T}}f(x), \max_{x \in \mathbb{T}} f(x)]\setminus \mathcal{D}(f)$, the number of zeros $N_0(f-\mu)$ (counting multiplicity) on $\mathbb{T}$ is {\it non-zero and constant} in $\mu$ on $[\mu_1, \mu_2]$.

Using an argument similar in spirit to what was used to prove Lemma \ref{DK}, we show in Appendix \ref{app_zerosets} that the following quantitative estimate holds.

\begin{prop} \label{obs_constantnumberzeros}
For any interval $[\mu_1, \mu_2]\subset [\min_{x\in \mathbb{T}}f(x), \max_{x\in \mathbb{T}}f(x)]\setminus\mathcal{D}(f)$, there exists $0 < R_0 = R_0(f, \mu_1, \mu_2)$ such that for every $0 < \delta \leq R_0$
\begin{equation}
\min_{x \in \mathbb{T}, \delta/2 \leq y \leq \delta} \vert f(x + i y)-\mu \vert \geq \hat{\beta}(f) \left(\frac{\delta}{2}\right)^{N_0(f-\mu_1)}  ~\mbox{,}
\end{equation}
where $g_\mu$ is the zero free part of $f-\mu$ on $\mathbb{T}$ and 
\begin{equation}\label{eq_betahat}
\hat{\beta}(f):= \min_{\mu \in [\mu_1, \mu_2]}\ \min_{z \in \mathbb{T}_{R_0}} \vert g_\mu(z) \vert.
\end{equation}
\end{prop}
\begin{remark}
In (\ref{eq_example_radius}), we explicitly quantify $R_0$ in terms of the properties of the family $\{f- \mu, \mu \in [\mu_1, \mu_2]\}$; for further details see Appendix \ref{app_zerosets}.
\end{remark}
We mention that the minimum taken with respect to $\mu$ in (\ref{eq_betahat}) exists by lower-semicontinuity of the map $\mu \mapsto \min_{z \in \mathbb{T}_{R_0}} \vert g_\mu(z) \vert$, see e.g. Proposition 4.2 in \cite{KleinDuarte_2013_posLE}.

Hence, under the hypotheses of Proposition \ref{obs_constantnumberzeros}, we can replace Observation \ref{obs_key} and Lemma \ref{DK} by Proposition \ref{obs_constantnumberzeros}. Then, the same arguments as in Sec. \ref{thm_main} imply that given $\lambda> 0$ large enough to guarantee
\begin{equation} \label{eq_guarantee}
\delta = \lambda^{-\frac{2}{2 N_0 + 1}} < R_0 ~\Leftrightarrow~ \lambda > R_0^{-\frac{2N_0 + 1}{2}} ~\mbox{,}
\end{equation}
we obtain the analogue of (\ref{eq_LEOFFreal}) for $E \in [\mu_1 \lambda , \mu_2 \lambda]$:
\begin{equation} \label{eq_LEOFFreal_1}
L(\alpha, A_{\lambda, E; y}) = \log{\lambda} + \int_\mathbb{T} \log \vert E/\lambda - f(x + i y) \vert ~\ud x + L(\alpha, D_{\lambda, E; y}) ~\mbox{,}
\end{equation}
with however the difference that (\ref{eq_LEOFFreal_1}) holds for {\em{all}} $y \in (\delta/2, \delta)$, independent of $E\in [\mu_1 \lambda, \mu_2 \lambda]$.

Then, since
\begin{equation*}
\lambda \vert {E}/{\lambda} - f(x + i y) \vert \geq  2^{-N_0} \hat{\beta}(f)\lambda^{\frac{1}{2 N_0 + 1}} ~\mbox{,}
\end{equation*}
Lemma \ref{lem_ds} implies that $(\alpha, D_{\lambda, E; y})$ induces a dominated splitting provided that 
\begin{equation}\label{eq_guarantee2}
2^{-N_0} \hat{\beta}(f) \lambda^{\frac{1}{2 N_0 + 1}} \geq 2  ~\Leftrightarrow~ \lambda > (2^{N_0}\hat{\beta}(f))^{2N_0+1} ~\mbox{,}
\end{equation}
which, and using (\ref{eq_guarantee}), is guaranteed if
\begin{equation} \label{eq_lambdatilde}
\lambda> \max \{ 2^{N_0}\hat{\beta}(f) ~,~ R_0^{-{1}/{2}} \}^{2 N_0 + 1} =: \tilde{\lambda}_0(f,\mu_1,\mu_2) ~\mbox{.}
\end{equation}

In particular, for $\vert \lambda \vert > \tilde{\lambda}_0$, Remark \ref{rem_ds} then yields 
\begin{equation*}
L(\alpha, D_{\lambda, E; y}) = O(\lambda^{-\frac{2}{2 N_0 + 1}}) ~\mbox{.}
\end{equation*}
Finally, from (\ref{eq_LEOFFreal_1}), smoothness of the LE locally about dominated splittings implies 
\begin{equation}
0 \leq \omega(\alpha, \lambda, E; y) = \omega[E/\lambda  - f](y) = \frac{1}{2} N_0 ~\mbox{, for all } y \in [ \delta/2 , \delta] ~\mbox{,}
\end{equation}
which gives rise to the desired improved estimate of the acceleration at $y=0$:
\begin{equation} \label{eq_improvedaccel}
\omega(\alpha, \lambda, E; 0) \leq \frac{1}{2} N_0(f - \mu_1) ~\mbox{.}
\end{equation}
The acceleration at $y=0$ can hence simply be estimated by counting the intersections of the graph of $f$ with horizontal lines at heights $\mu = \mu_1$.

In conclusion, relying on convexity of the complexified LE, the same argument as in the end of Sec. \ref{sec_proofthmmain} hence leads to the following stratified analogue of Theorem \ref{thm_main}.
\begin{theorem} \label{prop_strengthned}
Given a quasi-periodic Schr\"odinger operator (\ref{eq_schrodop}) with $f \in \mathcal{C}_h^\omega(\mathbb{T}; \mathbb{R})$ and irrational $\alpha$. 
For any interval $[\mu_1, \mu_2] \subset [\min_{x\in \mathbb{T}}f(x), \max_{x\in \mathbb{T}}f(x)]\setminus \mathcal{D}(f)$, there exist $0 < \tilde{\lambda}_0 = \tilde{\lambda}_0(f, \mu_1, \mu_2)$ (quantified in (\ref{eq_lambdatilde})) such that for all $\lambda> \tilde{\lambda}_0$ and $E\in [\mu_1 \lambda, \mu_2  \lambda ]$, one has
\begin{equation} \label{eq_LEOFFreal_1prime}
L(\alpha, A_{\lambda, E}) = \log{\lambda}+ \int_\mathbb{T} \log \vert E/\lambda - f(x) \vert ~\ud x + O(\lambda^{-\frac{2}{2N_0+1}}) ~\mbox{.}
\end{equation}
Moreover, in this case, the acceleration satisfies the upper bound
\begin{equation}
0 \leq \omega(\alpha, \lambda, E; 0) \leq \frac{1}{2} N_0(f-\mu_1) ~\mbox{.}
\end{equation}
\end{theorem}

\appendix

\section{Proof of Lemma \ref{DK}} \label{app_DuarteKlein}
For completeness and the reader's convenience, we include the proof of Lemma \ref{DK} below:
\begin{proof}
Let us consider 
\begin{align}\label{Cfac}
|f(x+iy)-\mu|=\left| g_{2\rho}(f-\mu)(x+iy) \prod_{j=1}^{\tilde{N}_{2\rho}(f-\mu)} \frac{e^{2\pi i (x+iy)}-e^{2\pi iz_{\mu, j}}}{2e^{2\pi}+2} \right|,
\end{align}
where $|\mathrm{Im}z_{\mu,j}|< 2\rho$ for $1\leq j\leq \tilde{N}_{2\rho}(f-\mu)$.

Let 
\begin{align*}
I_j=[\mathrm{Im}z_{\mu,j}-\frac{\delta}{2\pi N_{2\rho}(f-\mu)},\ \mathrm{Im}z_{\mu,j}+\frac{\delta}{2\pi N_{2\rho}(f-\mu)}].
\end{align*}
Then $|\bigcup_{j=1}^{\tilde{N}_{2\rho}(f-\mu)} I_j|\leq \frac{\delta \tilde{N}_{2\rho}(f-\mu)}{\pi N_{2\rho}(f-\mu)}<\frac{\delta}{2}$. This implies there exists $y_{\mu,\delta}$ such that $\frac{\delta}{2}< y_{\mu,\delta}< \delta$ and 
\begin{align}\label{disyzj}
|y_{\mu,\delta}-\mathrm{Im}z_{\mu,j}|> \frac{\delta}{2\pi N_{2\rho}(f-\mu)}.
\end{align}
Combining (\ref{Cfac}) with (\ref{disyzj}), we have
\begin{align*}
|f(x+iy_{\mu,\delta})-\mu|\geq &\beta_{\rho}(f-\mu)\left(\frac{e^{-2\pi }\delta} {(2e^{2\pi} +2) N_{2\rho}(f-\mu)}\right)^{\tilde{N}_{2\rho}(f-\mu)} \\
\geq &\hat{\beta}_{\rho}(f)\left(\frac{e^{-2\pi} \delta}{(2e^{2\pi}+2)\hat{N}_{2\rho}(f)}\right)^{ \hat{N}_{2\rho}(f)}. 
\end{align*}
In particular, we quantified the constant $K_1$ in Lemma \ref{DK} as
\begin{equation} \label{eq_constantk1}
K_1 = \dfrac{e^{-2\pi}}{2e^{2\pi}+2} ~\mbox{.}
\end{equation}
\end{proof}

\section{Estimates on the zero set of real analytic functions} \label{app_zerosets}

Let $0 \not\equiv f \in \mathcal{C}_h^\omega(\mathbb{T})$, where without loss, we take $0 < h < 1$. Suppose $f$ has zeros on $\mathbb{T}$. Since the zeros of $f$ are isolated, there exists $R = R(f) $ such that for all $0 < \delta\leq R$, one can find $\eta = \eta(f, \delta)$ such that,
\begin{equation} \label{eq_isolatedzeros_estim}
\min_{x \in \mathbb{T}, \delta/2 \leq y \leq \delta} \vert f(x + i y) \vert > \eta ~\mbox{.}
\end{equation}
To purpose of the following is to quantify the dependencies of $R$ and $\eta$. 

To this end, let $x_1, \dots, x_k$ denote the distinct zeros of $f$ on $\mathbb{T}$ with associated multiplicities $n_1, \dots, n_k$. Expanding $f$ in a power series about $x_j$, for $\vert z - x_j \vert \leq h$, one can write,
\begin{eqnarray} \label{eq_localexpzeros}
f(z) = (z - x_j)^{n_j}(a_{n_j}^{(j)} + h_j(z) ) ~\mbox{, }
h_j(z) = \sum_{m=1}^\infty a_{n_j + m}^{(j)} (z - x_j)^m ~\mbox{,}
\end{eqnarray}
where $a_{m}^{(j)} := \frac{1}{m!} f^{(m)}(x_j)$, $m \in \mathbb{N}$. 

Cauchy estimates readily show that for $0 < r < h$,
\begin{equation*}
\sup_{z \in \overline{B_r(x_j)}} \vert h(z) \vert \leq \dfrac{\Vert f \Vert_h}{h^{n_j}} \dfrac{r/h}{1- r/h} ~\mbox{.}
\end{equation*}
Thus, using (\ref{eq_localexpzeros}), we obtain that $f(z) \neq 0$ on $0 < \vert z - x_j \vert \leq r$ if
\begin{equation*}
r < h \dfrac{\vert a_{n_j}^{(j)} \vert h^{n_j}}{\Vert f \Vert_h + \vert a_{n_j}^{(j)} \vert h^{n_j}} = h \theta_{f}(\vert a_{n_j}^{(j)} \vert h^{n_j}) ~\mbox{.} 
\end{equation*}
Here, for $c>0$, we set $\theta_c(x): = \frac{x}{c + x}$, $x \in [0, + \infty)$, and to simplify notation, we write $\theta_f:= \theta_{\Vert f \Vert_h}$. Note that $0 \leq \theta_c(x) \nearrow 1$, as $x \to + \infty$.

Defining the {\em{modulus of transversality}} of $f$ on $\mathbb{T}$ by
\begin{equation} \label{eq_modulusofTransv}
\tau(f): = \min_{1 \leq j \leq k} \frac{1}{j!} \vert f^{(n_j)}(x_j) \vert ~\mbox{,}
\end{equation}
and letting $N(f):= \sum_{j=1}^k n_j$, we conclude (using $0< h <1$) that 
\begin{eqnarray}
f(z) \neq 0 ~\mbox{, } z \in \bigcup_{j=1}^{k} \left( \overline{B_r(x_j)}\setminus \{x_j\} \right) ~\mbox{, if } \\  \nonumber
r < h \theta_{f}\left( \tau(f) h^{N(f)} \right) =: \zeta(f) ~\mbox{.}
\end{eqnarray}

For arbitrary $x \in \mathbb{T} \setminus \left( \cup_{j=1}^{k} \overline{B_{\zeta(f)}(x_j)} \right)$, we follow a similar strategy, writing
\begin{equation*}
f(z) = f(x) + h_{x_0}(z) ~\mbox{, $\vert z - x \vert \leq h$ ,}
\end{equation*}
in which case Cauchy estimates yield
\begin{equation*}
\sup_{z \in \overline{B_r(x)}} \vert h_{x}(z) \vert \leq \Vert f \Vert_h \dfrac{r/h}{1 - r/h} ~\mbox{, for $0<r<h$ .}
\end{equation*} 
Since $f(x) \neq 0$, this in turn implies
\begin{equation}  \label{eq_radii_1}
f(z) \neq 0 ~\mbox{, } z \in \overline{B_r(x)} ~\mbox{, if } r < h \theta_f( \vert f(x) \vert )  ~\mbox{.}
\end{equation}

To obtain a uniform lower bound for $\vert f(x) \vert$ on $\mathbb{T} \setminus \left( \cup_{j=1}^{k} \overline{B_{\zeta(f)}(x_j)} \right)$, we factorize $f$ on $\mathbb{T}$ according to its zeros,
\begin{equation} \label{eq_factoriz_zeroslemma}
f(x) = g(x) \prod_{j=1}^{k} (x - x_j)^{n_j} ~\mbox{, }
\end{equation}
where $g$ is the zero-free part of $f$ on $\mathbb{T}$. Letting,
\begin{equation*}
\beta(f):= \min_{x \in \mathbb{T}} \vert g(x) \vert ~\mbox{,}
\end{equation*}
we therefore conclude,
\begin{equation} \label{eq_radii_2}
\vert f(x) \vert \geq \beta(f) \zeta(f)^{N(f)} ~\mbox{, all } x \in \mathbb{T} \setminus \left( \cup_{j=1}^{k} \overline{B_{\zeta(\mu)}(x_j)} \right) ~\mbox{.}
\end{equation}

Combining (\ref{eq_radii_1}) and (\ref{eq_radii_2}) hence shows that,
\begin{eqnarray}
f(z) \neq 0 ~\mbox{, } z \in \bigcup_{x \in \mathbb{T} \setminus \left( \cup_{j=1}^{k} \overline{B_{\zeta(f)}(x_j)} \right)} \overline{B_r(x)} ~\mbox{, if } \\ \nonumber
r < h \theta_f\left(\beta(f) \zeta(f)^{N(f)}\right) =: \gamma(f) ~\mbox{.}
\end{eqnarray}

In summary, simple geometric considerations to take into account the overlap of circles imply that $R(f)$ can be chosen according to
\begin{equation} \label{eq_appzero_radius}
R(f)= \delta(f)  \cdot \sin\left( \arctan\left( \dfrac{\gamma(f)}{\zeta(f)} \right)\right) = \dfrac{\zeta(f) \gamma(f)}{\sqrt{\zeta(f)^2 + \gamma(f)^2}} ~\mbox{.}
\end{equation}
Finally, because (\ref{eq_factoriz_zeroslemma}) extends to $\mathbb{T}_{R(f)}$, we can choose $\eta= \eta(f,\delta)$ in (\ref{eq_isolatedzeros_estim}) as
\begin{equation} \label{eq_appzero_lowerbound}
\eta(f,\delta) =(\delta/2)^{N(f)} \cdot \min_{z \in \mathbb{T}_\delta} \vert g(z) \vert   ~\mbox{.}
\end{equation}

\subsection*{Applications to Proposition \ref{obs_constantnumberzeros}} 
We apply the results of the previous section to the members of the family $f_\mu: = f - \mu$. As observed earlier, for $\mu \in [-\min_{x \in \mathbb{T}} f(x), \max_{x \in \mathbb{T}} f(x)]=:J(f)$, one has $N_0(f_\mu) >0$. Since $\eta_0$ in Proposition \ref{obs_constantnumberzeros} is quantified as an immediate consequence of (\ref{eq_appzero_lowerbound}), we focus on $R_0$.

Clearly, to quantify $R_0$ it suffices to obtain a uniform lower bound for $R(f_\mu)$ for $\mu \in [\mu_1, \mu_2]$. We mention that for general compact subsets of $J(f)$ the obstacle here is that the modulus of transversality $\tau(f_\mu)$ defined in (\ref{eq_modulusofTransv}) is not lower-semicontinuous in $\mu$; in fact, $\tau(f_\mu)$ may drop to zero if zeros collapse.

By the hypotheses in Proposition \ref{obs_constantnumberzeros} this is however prohibited and
\begin{equation}
\inf_{\mu \in [\mu_1, \mu_2]} \tau(f_\mu) \geq \min_{x \in f^{-1}([\mu_1, \mu_2])} \vert f^\prime(x) \vert =: \tau_0 ~\mbox{.}
\end{equation}

Therefore, if we set 
\begin{equation*}
M_0:= \max_{\mu \in [\mu_1, \mu_2]} \Vert f_\mu \Vert_h ~\mbox{, }
\beta_0:= \min_{\mu \in [\mu_1, \mu_2]} \beta(f_\mu) ~\mbox{,}
\end{equation*}
one has for all $\mu \in [\mu_1, \mu_2]$ that
\begin{eqnarray}
\zeta(f_\mu) \geq \zeta_0:= h \theta_{M_0}(h^{N_0} \tau_0) ~\mbox{,} \\
\gamma(f_\mu) \geq \gamma_0:= h \theta_{M_0}(h^{N_0} \beta_0) ~\mbox{.}
\end{eqnarray}
Thus, using (\ref{eq_appzero_radius}), we can choose $R_0$ as
\begin{equation} \label{eq_example_radius}
R_0 = \dfrac{\zeta_0 \gamma_0}{\sqrt{\zeta_0^2 + \gamma_0^2}} ~\mbox{.}
\end{equation}

\section{Numerical values for the constants in (\ref{eq_LEquantitativemodulus})} \label{app_numerical}

For simplicity we write $\sigma = \sigma(g)$ and $m = m(g)$. Based on the definition of $\sigma$ in (\ref{eq_estimatesolutionFP}), we distinguish between the cases $m \geq \frac{\sqrt{5}+3}{2}$ and $2 < m < \frac{\sqrt{5}+3}{2}$. Set $k = \frac{\sqrt{5}+3}{2}$. We will use the bounds
\begin{eqnarray}
\log(1 + x) \leq &  x & \mbox{, $x > -1$ ,}    \label{eq_logestim1} \\
\log(1 - x ) \geq & \dfrac{- x}{1 - x} & \mbox{, $0 < x  < 1$ .}  \label{eq_logestim2}
\end{eqnarray}

\begin{itemize}
\item{Case 1: $2 < m < k$:} In this case, $\sigma = 1$, whence using the upper bound in (\ref{eq_complexLE_ds_bd_1}) and (\ref{eq_logestim1}),
\begin{eqnarray}
L(\alpha, D) \leq \frac{1}{2} \log \left\{ (1 + \frac{1}{m} )^2 + \frac{1}{m^2} \right\} \leq \frac{1}{2} \log \left(1 + \frac{4}{m^2} \right) \leq \frac{2}{m^2}
\end{eqnarray}

On the other hand, the lower bound in (\ref{eq_complexLE_ds_bd_1}) together with (\ref{eq_logestim2}) implies
\begin{eqnarray}
L(\alpha, D) & \geq & \frac{1}{2} \log \left\{  \dfrac{ ( 1 - \frac{1}{m^2} )^2 + \frac{1}{m^2} }{2}     \right\} \geq -\frac{1}{2} \left[ \dfrac{2(m-1)}{m^2 - 2 (m-1)} + \log{2} \right]  \\
& \geq & - \frac{1}{m^2} \sup_{2 < m \leq k}  \left[ \dfrac{2(m-1)m^2}{2(m^2 - 2 (m-1))} + \frac{\log{2}}{2} m^2 \right] =: - \frac{c}{m^2} ~\mbox{.}
\end{eqnarray}
Since the argument of the supremum increases in $m$, it is attained at $m=k$, which determines $c \approx 5.4407$. In summary, for $2 < m < k$, we conclude
\begin{equation}
- \frac{c}{m^2} \leq L(\alpha, D) \leq \frac{2}{m^2} ~\mbox{.}
\end{equation}

\item{Case 2: $m \geq k$:} First observe that from the definition of $\sigma$, $m > k$ and $m \mapsto \frac{m-1}{m-2}$ decreasing gives
\begin{equation}
\frac{1}{m} \leq \sigma \leq \left( \frac{\sqrt{5} + 1}{\sqrt{5} - 1} \right) \frac{1}{m} =:  \frac{c_+}{m} ~\mbox{.}
\end{equation}

Thus, from (\ref{eq_complexLE_ds_bd_1}) and (\ref{eq_logestim1})-(\ref{eq_logestim2}), we estimate
\begin{eqnarray}
2 L(\alpha, D) & \geq & \log \left[ 1 - \left( \frac{2 c_+ - 1}{m^2} - \frac{1}{m^4} \right) \right] - \log\left( 1+ \frac{c_+^2}{m^2} \right) \nonumber \\
                   & \geq & - \frac{1}{m^2} \left\{ c_+^2 + \dfrac{ (2 c_+ - 1) m^4 - m^2  }{  m^4 - (2 c_+ -  1) m^2 + 1 }   \right\} =: - \frac{1}{m^2} \left\{ c_+^2 + f(m) \right\}
\end{eqnarray}
Since $f(m)$ decreases in $m > 2$, we thus obtain
\begin{equation}
L(\alpha, D) \geq - \frac{d_-}{m^2} ~\mbox{, } d_- := \frac{1}{2} ( c_+^2 + f(k) ) \approx 8.4985
\end{equation}

In view of the upper bound, (\ref{eq_complexLE_ds_bd_1}) and (\ref{eq_logestim1}) yields
\begin{eqnarray}
L(\alpha, D) & \leq & \frac{1}{2} \log \left( 1 + \frac{2 c_+ + 1}{m^2} + \frac{c_+^2}{m^4} \right) \\
             & \leq & \frac{1}{2} \log  \left(1 + \frac{(c_+ + 1)^2}{m^2} \right) \leq \frac{(c_+ + 1)^2}{2 m^2} =: \frac{d_+}{m^2} ~\mbox{,}
\end{eqnarray}
where $d_+ \approx 6.5451$. 
\end{itemize}

In summary, combining the two cases, for all $m > 2$, we obtain the estimate (\ref{eq_LEquantitativemodulus}) with
\begin{eqnarray}
K_2 & := & \max\{ d_- ~, ~ c\} = d_- \approx 8.4985 ~\mbox{, } \\
K_3 & := & \max\{ d_+ ~, ~ 2\} = d_+  \approx 6.5451 ~\mbox{,}
\end{eqnarray}
as claimed in (\ref{eq_quantificConst}).

\section*{Acknowledgement}
R. Han was partially supported by the NSF DMS–1401204.
We would like to thank Svetlana Jitomirskaya for useful discussions.
We also gratefully acknowledge support from the Simons Center for Geometry and Physics, Stony Brook University, where this work was started.

\bibliographystyle{amsplain}

\begin{thebibliography}{10}

\bibitem{Avila_globalthy_published} A. Avila, \textit{Global theory of one-frequency Schr\"odinger operators}, Acta Math. 215, 1-54 (2015).

\bibitem{Bourgain_book_2005} J. Bourgain, \textit{Green's function estimates for lattice Schr單inger operators and applications}, Princeton Univ. Press (2005), Princeton.

\bibitem{Damanik_review2014} D. Damanik, \textit{Schr\"odinger Operators with Dynamically Defined Potentials: A Survey}, Ergodic Theory and Dynamical Systems, to appear.

\bibitem{KleinDuarte_2013_posLE} P. Duarte and S. Klein, \textit{Positive Lyapunov exponents for higher dimensional quasiperiodic cocycles}, Commun. Math. Phys. {\bf 332} 189-219 (2014)

\bibitem{DuarteKlein_2015_contising} P. Duarte and S. Klein, \textit{Continuity, positivity, and simplicity of the Lyapunov exponents for quasi-periodic cocycles}, preprint (2015); available on arXiv:1603.06851 [math.DS]

\bibitem{DuarteKlein_monograph} P. Duarte and S. Klein, \textit{Lyapunov exponents of linear cocycles}, Atlantis Press (2016), Amsterdam.


\bibitem{GeYou}L. Ge, J. You, \textit{Applications of Quantitative Almost Reducibility}, a talk during the workshop ``Between Dynamics and Spectral Theory'' June 6-10, 2016 at Simons Center for Geometry and Physics, Stony Brook University.

\bibitem{Herman_1983} M. Herman, \textit{Une methode pour minorer les exposants des Lyapunov et quelques examples montrant le charact\`ere local d'un th\'eor\`eme d'Arnold et de Moser sur le tore de dimension 2}, Comment. Math. Helv. 58 (1983), 453 - 562.


\bibitem {JitomirskayaMarx_JFPT_2011} S. Jitomirskaya, C.A. Marx, \textit{Continuity of the Lyapunov Exponent for analytic quasi-periodic cocycles with singularities}, J. Fixed Point Theory Appl. 10 (2011), 129-146. 

\bibitem {JitomirskayaMarx_ETDS_2016_review} S. Jitomirskaya, C.A. Marx, \textit{Dynamics and spectral theory of quasi-periodic Schr\"odinger-type operators}, Ergodic Theory and Dynamical Systems, to appear.

\bibitem{Marx_Shou_Wellens_2015} C. A. Marx, L. Shou, and J. Wellens, \textit{Subcritical behavior for quasi-periodic Schr\"odinger cocycles with trigonometric potentials}, to appear in the Journal of Spectral Theory (2015). Preprint vailable on arXiv:1509.05279v2 [math-ph].

\bibitem{Ruelle_1979_LEanalytic} D. Ruelle, \textit{Analyticity Properties of the Characteristic Exponents of Random Matrix Products}, Advances in Mathematics 32 (1979), 68-80.



\bibitem{SoretsSpencer_1991} E. Sorets, T. Spencer, \textit{Positive Lyapunov Exponents for Schr\"odinger Operators with quasi-periodic potentials}, CMP 142 (3), 543 - 566 (1991).


\bibitem{Zhang_2012} Z. Zhang, \textit{Positive lyapunov exponents for quasiperiodic Szeg\H{o} cocycles}, Nonlinearity 25 (2012), 1771 - 1797.



\end{thebibliography}

\end{document}